\documentclass[english]{cccconf}
\usepackage{epstopdf}

\usepackage{amsmath,amssymb,amsfonts}
\usepackage{amsthm}
\usepackage[ruled,linesnumbered]{algorithm2e}
\usepackage{array}
\usepackage{booktabs}
\usepackage{color}
\usepackage{cite}
\usepackage{cuted}
\usepackage{enumitem}
\usepackage{graphicx}
\usepackage{hyperref}
\usepackage{mathrsfs}
\usepackage{multirow}
\usepackage[caption=false,font=scriptsize,labelfont=sf,textfont=sf]{subfig}
\usepackage{textcomp}
\usepackage{threeparttable}
\usepackage{stfloats}
\usepackage{upgreek}
\usepackage{url}
\usepackage{verbatim}

\newcommand\dif{\mathrm{d}}
\newtheorem{definition}{Definition}
\newtheorem{theorem}{Theorem}

\begin{document}

\title{Optimal Investment with Herd Behaviour Using Rational Decision Decomposition}

\author{Huisheng Wang\aref{thu} and H. Vicky Zhao\aref{thu}}

\affiliation[thu]{Department of Automation,
        Tsinghua University, Beijing 100084, P.~R.~China
        \email{\href{whs22@mail.tsinghua.edu.cn}{whs22@mail.tsinghua.edu.cn}, \href{vzhao@tsinghua.edu.cn}{vzhao@tsinghua.edu.cn}}}
        
\maketitle

\begin{abstract}
In this paper, we study the optimal investment problem considering the herd behaviour between two agents, including one leading expert and one following agent whose decisions are influenced by those of the leading expert. In the objective functional of the optimal investment problem, we introduce the average deviation term to measure the distance between the two agents' decisions and use the variational method to find its analytical solution. To theoretically analyze the impact of the following agent's herd behaviour on his/her decision, we decompose his/her optimal decision into a convex linear combination of the two agents' rational decisions, which we call the rational decision decomposition. Furthermore, we define the weight function in the rational decision decomposition as the following agent's investment opinion to measure the preference of his/her own rational decision over that of the leading expert. We use the investment opinion to quantitatively analyze the impact of the herd behaviour, the following agent's initial wealth, the excess return, and the volatility of the risky asset on the optimal decision. We validate our analyses through numerical experiments on real stock data. This study is crucial to understanding investors' herd behaviour in decision-making and designing effective mechanisms to guide their decisions.
\end{abstract}


\section{Introduction}
There exists a category of leading experts in the financial markets, such as financial analysts, fund managers, or other professionals, who share their decisions through channels like social media and websites \cite{bodnaruk2015financial}. As these leading experts' decisions are often considered valuable investment guides, investors, referred to as agents below, tend to mimic their decisions \cite{brown2008neighbors}. Additionally, agents may have an intrinsic preference for conformity, and lean towards the decisions of the other agents to build consensus and trust \cite{hong2004social}. This phenomenon is commonly known as herd behaviour in behavioural finance \cite{eguiluz2000transmission, alfarano2005estimation, hu2016study}. Prior studies have qualitatively verified the existence of herd behaviour in agents' decisions in financial markets \cite{chang2000examination, chong2017explains, yu2019detection, wanidwaranan2022unintentional, zhou2022internet}. However, to quantitatively analyze the influence of herd behaviour on agents' decisions, it is necessary to employ mathematical models based on the optimal investment theory.

Optimal investment theory provides a theoretical framework to study how agents dynamically adjust their investment decisions throughout the investment period to maximize the return and minimize the volatility using the stochastic optimal control theory \cite{1952Portfolio, samuelson1975lifetime, calafiore2008multi, yang2019risk, li2022kind, aljalal2022optimal, alasmi2023optimal}. The Merton problem is one of the classical optimal investment problems \cite{merton1969lifetime}, where the agent allocates a certain amount of money to a risky asset at each moment to maximize the expected utility of the terminal wealth, i.e., the wealth at the end of the investment period. There have been many quantitative studies analyzing the impact of interactions among agents on the optimal investment. The previous studies in \cite{gali1994keeping, gollier2004misery, lauterbach2004keeping, gomez2007impact, gomez2009implications, roussanov2010diversification, garcia2011relative, levy2015keeping} examined the impact of relative wealth and consumption among agents on their decisions, where the objective was a functional of each agent's excess wealth or consumption relative to other agents during the investment period. However, to the best of our knowledge, few studies have investigated how herd behaviour affects agents' decisions on the allocation of risky assets in the Merton problem.

In this paper, we study an optimal investment problem involving two agents, one leading expert, and one following agent whose decisions are influenced by those of the leading expert, and we quantitatively analyze the impact of herd behaviour on the following agent's decision. The prior work in \cite{merigo2011induced} used the traditional Euclidean distance to quantitatively measure the disparity between agents' decisions. In this paper, we incorporate an exponential decay term into the traditional Euclidean distance to reflect the fact that agents often give less weight to future rewards when making decisions \cite{frederick2002time}, and we call the modified Euclidean distance the average deviation. We introduce the average deviation into the objective functional of the traditional Merton problem and employ the variational method to analytically solve the stochastic optimal control problem and find the following agent's optimal decision.

Note that with the above analytical solution of the following agent's optimal decision, it is theoretically difficult to directly examine the influence of herd behaviour due to its complicated expression. To address this problem, we introduce a method called rational decision decomposition to transform the following agent's optimal decision into a convex linear combination of the two agents' rational decisions, i.e., the optimal decisions in the traditional Merton problem, where the agents make independent decisions without the influence of other agents. Furthermore, we define the weight function as the following agent's investment opinion, which indicates the extent to which the following agent's decision aligns with his/her own rational decision. We use the following agent's investment opinion to quantitatively analyze the impact of the herd behaviour, the following agent's initial wealth, the excess return, and the volatility of the risky asset on the optimal decision.

The rest of this paper is organized as follows. We introduce the average deviation to quantify the disparity between the two agents' decisions, and present the dual-agent optimal investment problem considering herd behaviour in Section \ref{sec:decision}. We derive the analytical solution of the following agent's optimal decision using the variational method in Section \ref{sec:solution}, and quantitatively analyze the impact of herd behaviour on the following agent's optimal decisions using the rational decision decomposition in Section \ref{sec:analysis}. In Section \ref{sec:experiment}, we conduct numerical experiments on a real stock dataset to validate our analyses. Section \ref{sec:conclusion} is the conclusion.

\section{Problem Formulation}
\label{sec:decision}

\subsection{The Traditional Merton Problem}
Following the prior work in \cite{merton1969lifetime}, we consider the problem of an agent investing in the period $\mathcal{T}:=[0, T]$ in a financial market consisting of a risk-free asset and a risky asset. Let $(\Omega,\mathscr{F},\{\mathscr{F}(t)\}_{t\in\mathcal{T}},\mathbb{P})$ be a complete filtered probability space on which a standard Brownian motion $\{W(t)\}_{t\in\mathcal{T}}$ is defined. The price process of the risk-free asset $\{B(t)\}_{t\in\mathcal{T}}$ and the risky asset $\{S(t)\}_{t\in\mathcal{T}}$ are given by
\begin{align}
    \dif B(t)&=rB(t) \dif t\quad\mbox{and}\\
    \dif S(t)&=\mu S(t)\dif t+\sigma S(t)\dif W(t),
\end{align}
respectively, where $r$ is the interest rate of the risk-free asset, and $\mu$ and $\sigma$ are the appreciation rate and the volatility of the risky asset. We define $v=\mu-r>0$ as the excess return rate of the risky asset. Let $x$ denote the agent's initial wealth, and let $P(t)$ denote the amount of the risky asset held by the agent at time $t$. According to \cite{merton1969lifetime}, the agent's wealth process $\{X(t)\}_{t\in\mathcal{T}}$ can be expressed as 
\begin{equation}
    \dif X(t)=[rX(t)+vP(t)]\dif t+\sigma P(t)\dif W(t),
    \label{eq:sde}
\end{equation}
subject to $X(0)=x$. In the traditional Merton problem, the agent determines $\{P(t)\}_{t\in\mathcal{T}}$ to maximize the expected utility of the terminal wealth $\mathbb{E}\phi(X(T))$. The utility function $\phi(X(T))$ satisfies the characteristics of diminishing marginal returns and concavity \cite{pratt1978risk}. In this paper, we consider the Constant Absolute Risk Aversion form \cite{pratt1978risk}:
\begin{equation}
    \phi(X(T))=-\frac{1}{\alpha}\mathrm{e}^{-\alpha X(T)},
\end{equation}
where $\alpha>0$ is referred to as the \textit{risk aversion coefficient}. 
As the risk aversion coefficient $\alpha$ increases, the utility function becomes more sensitive to changes in the terminal wealth. 

In summary, the traditional Merton problem becomes
\begin{align*}
    \text{\textbf{Problem 1.}}&\sup_{\{P(t)\}_{t\in\mathcal{T}}\in\mathcal{U}}\mathbb{E}\phi(X(T))\\
    \text{s.t.\quad}&\dif X(t)=[rX(t)+vP(t)]\dif t+\sigma P(t)\dif W(t),\\
    &X(0)=x,
\end{align*}
which is a stochastic optimal control problem. In Problem 1, $\mathcal{U}$ is the set of admissible decisions, which is a subset of $\mathcal{L}^1:=\left\{\{u(t)\}_{t\in\mathcal{T}} \middle|\mathbb{E}\int_0^T|u(t)|\dif t<+\infty\right\}$.

According to \cite{merton1969lifetime}, the agent's optimal decision $\{\bar{P}(t)\}_{t\in\mathcal{T}}$ in Problem 1 is given by
\begin{equation}
    \bar{P}(t)=\frac{v}{\alpha\sigma^2}\mathrm{e}^{r(t-T)},t\in\mathcal{T}.
    \label{eq:merton-optimal-decision}
\end{equation}
We call \eqref{eq:merton-optimal-decision} the agent's \textit{rational decision}. According to \eqref{eq:merton-optimal-decision}, the rational amount of the risky asset is proportional to the excess return rate $v$ and inversely proportional to the volatility $\sigma$ and the risk aversion coefficient $\alpha$.

\subsection{Optimal Investment with Herd Behaviour}
Given the traditional Merton problem, we study the optimal investment problem with consideration of herd behaviour. The problem involves two agents, A1 and A2, where A1 is the following agent and A2 is the leading expert. We denote $\alpha_1$ and $\alpha_2$ as their risk aversion coefficients, respectively. We assume that A2's decision follows the form of the rational decision $\{\bar{P}_2(t)\}_{t\in\mathcal{T}}$ as in \eqref{eq:merton-optimal-decision}, and A1's decision $\{P_1(t)\}_{t\in\mathcal{T}}$ is unilaterally influenced by A2's decision. Considering the herd behaviour, A1 aims to maximize his/her expected utility of the terminal wealth while minimizing the distance between his/her decision and A2's decision. The prior work in \cite{merigo2011induced} used the Euclidean distance to measure the distance between two decisions, i.e., $\frac{1}{2}\int_0^T[P_1(t)-\bar{P}_2(t)]^2\dif t$ in continuous-time format. In this paper, we introduce an exponential decay term $\mathrm{e}^{\rho r(T-t)}$ to reflect the fact that agents often give less weight to future rewards when making decisions, and we call the modified Euclidean distance the \textit{average deviation}.
\begin{definition}\label{def:1}
    The average deviation between A1 and A2's decisions $\{P_1(t)\}_{t\in\mathcal{T}}$ and $\{\bar{P}_2(t)\}_{t\in\mathcal{T}}$ is defined as:
\begin{equation}
    D[P_1\Vert \bar{P}_2]=\frac{1}{2}\int_0^T\mathrm{e}^{\rho r(T-t)}[P_1(t)-\bar{P}_2(t)]^2\dif t,
    \label{eq:herd-cost}
\end{equation}
where $\rho\geqslant0$ is the decay rate. 
\end{definition}
A higher decay rate implies that deviations occurring later carry less weight. When $\rho=0$, deviations for all times are equally weighted, and \eqref{eq:herd-cost} becomes the same as the traditional Euclidean distance in \cite{merigo2011induced}. Let $\{X_1(t)\}_{t\in\mathcal{T}}$ denote A1's wealth process. 

With Definition \ref{def:1}, we define the following optimal investment problem for A1's decision:
\begin{align*}
    \text{\textbf{Problem 2.}}&\sup_{\{P_1(t)\}_{t\in\mathcal{T}}\in\mathcal{U}}\mathbb{E}\phi(X_1(T))-\theta D[P_1\Vert \bar{P}_2]\\
    \text{s.t. \quad}&\dif X_1(t)=[rX_1(t)+vP_1(t)]\dif t+\sigma P_1(t)\dif W(t),\\
    &X_1(0)=x_1,
\end{align*}
where the herd coefficient $\theta>0$ is to address the tradeoff between the two different objectives, i.e., maximizing the expected utility of the terminal wealth $\mathbb{E}\phi(X_1(T))$ and minimizing the average deviation $D[P_1\Vert \bar{P}_2]$. When $\theta=0$, A1's optimal decision is entirely independent of A2's decision, and Problem 2 degenerates into the traditional Merton Problem, i.e., $\lim_{\theta\to0}P_1^*(t)=\bar{P}_1(t)$ for all $t\in\mathcal{T}$. As $\theta$ approaches infinity, A1's optimal decision is equal to A2's decision, i.e., $\lim_{\theta\to+\infty}P_1^*(t)=\bar{P}_2(t)$ for all $t\in\mathcal{T}$. 

\section{Solution to the Optimal Investment Problem}
\label{sec:solution}
Next, we use the variational method to solve Problem 2. For the convenience of expression, we define $\vartheta:=\frac{\theta}{\alpha_1\sigma^2}$ and $\varrho:=2-\rho$ as the modified herd coefficient and decay rate, respectively.

\begin{theorem}\label{the:1}
The optimal decision $\{P_1^*(t)\}_{t\in\mathcal{T}}$ in Problem 2 is given by
\begin{equation}
    P_1^*(t)=\frac{\eta\alpha_2\sigma^2\mathrm{e}^{\varrho r(T-t)}+\theta}{\eta\alpha_1\sigma^2\mathrm{e}^{\varrho r(T-t)}+\theta}\cdot \frac{v}{\alpha_2\sigma^2}\mathrm{e}^{r(t-T)},t\in\mathcal{T},
    \label{eq:optimal-decision-one}
\end{equation}
where $\eta$ is the integral constant calculated by Algorithm \ref{alg:1}.
\end{theorem}

\begin{algorithm}[t]
	\caption{Iterative Method of the Integral Constant $\eta$.}
	\label{alg:1}
	\KwIn{Interest rate: $r$; Excess return rate: $v$; Volatility: $\sigma$; Initial wealth $x_1$; Risk aversion coefficients: $\alpha_1,\alpha_2$; Investment period: $T$; Modified herd coefficient: $\vartheta$; Modified decay rate: $\varrho$; Tolerance: $\varepsilon$.}
	\KwOut{Integral constant: $\eta$.}  
	\BlankLine
	$\eta_0=\exp\left\{-\alpha_1x_1\mathrm{e}^{rT}-\frac{v^2T}{2\sigma^2}\right\}$, $\Delta \eta_0=+\infty$, $k=0$;

	\While{$\Delta\eta_k\geqslant\varepsilon$}{$\eta_{k+1}=\eta_0\exp\left\{\int_0^T\frac{\vartheta^2v^2\left(\alpha_1/\alpha_2-1\right)^2\dif t}{2\sigma^2\left[\eta_k\mathrm{e}^{\varrho r(T-t)}+\vartheta\right]^2}\right\}$;
 
$\Delta\eta_{k+1}=\left|\eta_{k+1}-\eta_k\right|$;
    
    $k\leftarrow k+1$;}
	$\eta\approx\eta_k$ ($\left|\eta-\eta_k\right|<\varepsilon$).
\end{algorithm}

\begin{proof}
First, by solving \eqref{eq:sde}, A1's terminal wealth is
\begin{equation}
    X_1(T)=x_1\mathrm{e}^{rT}+\int_0^T\mathrm{e}^{r(T-t)}P_1(t)[v\dif t+\sigma\dif W(t)].
\end{equation}
Therefore, the terminal wealth $X_1(T)$ is a normal random variable whose expectation and variance are
\begin{align}
    \mathbb{E}X_1(T)&=x_1\mathrm{e}^{rT}+v\int_0^T\mathrm{e}^{r(T-t)}P_1(t)\dif t\quad\mbox{and}\\
    \mathbb{D}X_1(T)&=\sigma^2\int_0^T\mathrm{e}^{2r(T-t)} P_1^2(t)\dif t,
\end{align}
respectively. Because the set of admissible decisions $\mathcal{U}$ is a subset of $\mathcal{L}^1$, we have $\mathbb{E}X_1(t)<\infty$ and $\mathbb{D}X_1(t)<\infty$. Because $\phi(X_1(T))$ is an exponential function of $X_1(T)$, it is a lognormal random variable whose expectation is
\begin{equation}
    \mathbb{E}\phi(X_1(T))=-\frac{1}{\alpha_1}\exp\left\{-\alpha_1\mathbb{E}X_1(T)+\frac{\alpha_1^2}{2}\mathbb{D}X_1(T)\right\}.
\end{equation}
Therefore, the objective functional of Problem 2 can be expressed as a functional of $\{P_1(t)\}_{t\in\mathcal{T}}$, i.e.,
\begin{align}
    J[P_1]&:=-\frac{1}{\alpha_1}\exp\left\{-\alpha_1x_1\mathrm{e}^{rT}-\alpha_1v\int_0^T\mathrm{e}^{r(T-t)}P_1(t)\dif t\right.\notag\\
    &+\left.\frac{\alpha_1^2\sigma^2}{2}\int_0^T\mathrm{e}^{2r(T-t)}P_1^2(t)\dif t\right\}\notag\\
    &-\frac{\theta}{2}\int_0^T\mathrm{e}^{\rho r(T-t)}[P_1(t)-\bar{P}_2(t)]^2\dif t.
\end{align}
According to the variational method, the first and second variations are
\begin{align}
    \delta J
    &=\exp\left\{-\alpha_1x_1\mathrm{e}^{rT}-\alpha_1v\int_0^T\mathrm{e}^{r(T-t)}P_1(t)\dif t\right.\notag\\
    &+\left.\frac{\alpha_1^2\sigma^2}{2}\int_0^T\mathrm{e}^{2r(T-t)}P_1^2(t)\dif t\right\}\notag\\
    &\cdot\int_0^T\left[v\mathrm{e}^{r(T-t)}-\alpha_1\sigma^2\mathrm{e}^{2r(T-t)}P_1(t)\right]\delta P_1(t)\dif t\notag\\
    &-\theta\int_0^T\mathrm{e}^{\rho r(T-t)}[P_1(t)-\bar{P}_2(t)]\delta P_1(t)\dif t\quad\mbox{and}
\end{align}
\begin{align}
    \delta^2J
    &=-\exp\left\{-\alpha_1x_1\mathrm{e}^{rT}-\alpha_1v\int_0^T\mathrm{e}^{r(T-t)}P_1(t)\dif t\right.\notag\\
    &+\left.\frac{\alpha_1^2\sigma^2}{2}\int_0^T\mathrm{e}^{2r(T-t)}P_1^2(t)\dif t\right\}\cdot\alpha_1\left\{\int_0^T\left[v\mathrm{e}^{r(T-t)}\right.\right.\notag\\
    &-\left.\left.\alpha_1\sigma^2\mathrm{e}^{2r(T-t)}P_1(t)\right]\delta P_1(t)\dif t\right\}^2\notag\\
    &-\int_0^T\left[\exp\left\{-\alpha_1x_1\mathrm{e}^{rT}-\alpha_1v\int_0^T\mathrm{e}^{r(T-t)}P_1(t)\dif t\right.\right.\notag\\
    &+\left.\left.\frac{\alpha_1^2\sigma^2}{2}\int_0^T\mathrm{e}^{2r(T-t)}P_1^2(t)\dif t\right\}\cdot\alpha_1\sigma^2\mathrm{e}^{2r(T-t)}\right.\notag\\
    &+\left.\theta\mathrm{e}^{\rho r(T-t)}\right][\delta P_1(t)]^2\dif t<0.
    \label{eq:2delta}
\end{align}
According to the Jacobi condition, the sufficient and necessary condition for the functional supremum is given by $\delta J=0$ for $\{P_1^*(t)\}_{t\in\mathcal{T}}$. Therefore, we have
\begin{equation}
    P_1^*(t)=\frac{\eta\alpha_2\sigma^2\mathrm{e}^{\varrho r(T-t)}+\theta}{\eta\alpha_1\sigma^2\mathrm{e}^{\varrho r(T-t)}+\theta}\cdot\bar{P}_2(t),t\in\mathcal{T}.
    \label{eq:optimal-decision-one-P2}
\end{equation}
Given $\{\bar{P}_2(t)\}_{t\in\mathcal{T}}$ as in \eqref{eq:merton-optimal-decision}, we have \eqref{eq:optimal-decision-one}, where the integral constant is defined as
\begin{align}
    \eta&:=-\alpha_1\mathbb{E}\phi(X_1^*(T))\notag\\
    &=\exp\left\{-\alpha_1x\mathrm{e}^{rT}-\alpha_1v\int_0^T\mathrm{e}^{r(T-t)}P_1^*(t)\dif t\right.\notag\\
    &+\left.\frac{\alpha_1^2\sigma^2}{2}\int_0^T\mathrm{e}^{2r(T-t)}P_1^{*2}(t)\dif t\right\}.
    \label{eq:int}
\end{align}
From \eqref{eq:int}, $\eta$ is proportional to the utility of the optimal terminal wealth $X_1^*(T)$, the terminal wealth under the optimal decision $\{P_1^*(t)\}_{t\in\mathcal{T}}$, and the weight is $-\alpha_1<0$. Therefore, the integral constant $\eta$ is a decreasing function of the expected utility of the terminal wealth $\mathbb{E}\phi(X_1(T))$. By substituting \eqref{eq:optimal-decision-one} into \eqref{eq:int}, we have
\begin{equation}
    \eta=\exp\left\{-\alpha_1x_1\mathrm{e}^{rT}-\frac{v^2T}{2\sigma^2}+\int_0^T\frac{\vartheta^2v^2\left(\alpha_1/\alpha_2-1\right)^2\dif t}{2\sigma^2\left[\eta \mathrm{e}^{\varrho r(T-t)}+\vartheta\right]^2}\right\}.
    \label{eq:eta-equation}
\end{equation}
However, \eqref{eq:eta-equation} has no closed-form solution. Thus, we propose Algorithm 1 to calculate its numerical solution.
\end{proof}

According to the Fixed Point Theorem, we prove the convergence of Algorithm \ref{alg:1} in Theorem \ref{the:2}. 

\begin{theorem}
\label{the:2}
Algorithm \ref{alg:1} converges to $\eta$ if
\begin{equation}
    \frac{\vartheta^2v^2\left(\alpha_1-\alpha_2\right)^2\overline{\eta}}{\alpha_2^2\sigma^2\underline{\eta}^3}\cdot\frac{1-\mathrm{e}^{-2\varrho rT}}{2\varrho r}\leqslant1,
    \label{eq:convergence-domain}
\end{equation}
where the lower bound is $\underline{\eta}=\exp\left\{-\alpha_1x_1\mathrm{e}^{rT}-\frac{v^2T}{2\sigma^2}\right\}$, and the upper bound is $\overline{\eta}=\underline{\eta}\exp\left\{\int_0^T\frac{\vartheta^2v^2\left(\alpha_1/\alpha_2-1\right)^2\dif t}{2\sigma^2\left[\underline{\eta}\mathrm{e}^{\varrho r(T-t)}+\vartheta\right]^2}\right\}$.
\end{theorem}

\begin{proof}
We define the following iteration function:
\begin{equation}
    f(\xi):=\underline{\eta}\exp\left\{\int_0^T\frac{\vartheta^2v^2\left(\alpha_1/\alpha_2-1\right)^2\dif t}{2\sigma^2\left[\xi \mathrm{e}^{\varrho r(T-t)}+\vartheta\right]^2}\right\},
    \label{eq:iter}
\end{equation}
which satisfies $f(\eta)=\eta$, i.e., $\eta$ is a fixed point of the iteration function $f$. To prove the convergence, we need to find a self-mapping interval, and then prove that the iteration function satisfies the Lipschitz continuity condition on the self-mapping interval. First, it is obvious that $f(\xi)\geqslant\underline{\eta}$ for all $\xi\in[\underline{\eta},\overline{\eta}]$ because the exponential term in \eqref{eq:iter} is greater than $1$. Second, it can be easily proved that $f(\xi)$ is a monotonically decreasing function with respect to $\xi$ for all $\xi\in[\underline{\eta},\overline{\eta}]$, so we have $f(\xi)\leqslant f(\underline{\eta})=\overline{\eta}$ for all $\xi\in[\underline{\eta},\overline{\eta}]$. Therefore, there is always $f(\xi)\in[\underline{\eta},\overline{\eta}]$ for all $\xi\in[\underline{\eta},\overline{\eta}]$, i.e., $f:[\underline{\eta},\overline{\eta}]\mapsto [\underline{\eta},\overline{\eta}]$.
It should be noted that $f(\xi)$ is a continuously differentiable function on $[\underline{\eta},\overline{\eta}]$. To check the Lipschitz continuity condition, note that
\begin{align}
    |f'(\xi)|
    &=f(\xi)\int_0^T\frac{\vartheta^2v^2\left(\alpha_1/\alpha_2-1\right)^2\mathrm{e}^{\varrho r(T-t)}\dif t}{\sigma^2\left[\xi\mathrm{e}^{\varrho r(T-t)}+\vartheta\right]^3}\notag\\
    &<f(\xi)\int_0^T\frac{\vartheta^2v^2\left(\alpha_1/\alpha_2-1\right)^2\dif t}{\xi^3\sigma^2\mathrm{e}^{2\varrho r(T-t)}}\notag\\
    &=\frac{\vartheta^2v^2\left(\alpha_1-\alpha_2\right)^2f(\xi)}{\xi^3\alpha_2^2\sigma^2}\cdot\frac{1-\mathrm{e}^{-2\varrho rT}}{2\varrho r}.
\end{align}
If \eqref{eq:convergence-domain} holds, then we have
\begin{equation}
    |f'(\xi)|<\frac{\vartheta^2v^2\left(\alpha_1-\alpha_2\right)^2\overline{\eta}}{\alpha_2^2\sigma^2\underline{\eta}^3}\cdot\frac{1-\mathrm{e}^{-2\varrho rT}}{2\varrho r}\leqslant1
\end{equation}
for all $\xi\in[\underline{\eta},\overline{\eta}]$. Therefore, Algorithm \ref{alg:1} converges when the initial value $\xi\in[\underline{\eta},\overline{\eta}]$. Because the initial value in Algorithm \ref{alg:1} is $\eta_0=\underline{\eta}\in[\underline{\eta},\overline{\eta}]$, Algorithm \ref{alg:1} converges to $\eta$.
\end{proof}

\section{Impact of Herd Behaviour on Optimal Decision}
\label{sec:analysis}
\subsection{Rational Decision Decomposition and Investment Opinion}
Theorem \ref{the:1} and Algorithm \ref{alg:1} provide the analytical solution for A1's optimal decision in Problem 2. However, due to the complexity of the expression of \eqref{eq:optimal-decision-one}, it is complicated to theoretically analyze the impact of herd behaviour on A1's optimal decision. Because the objective functional in Problem 2 is a weighted average of the expected utility of terminal wealth $\mathbb{E}\phi(X_1(T))$ and the average deviation $D[P_1\Vert\bar{P}_2]$, an intuitive notion is that A1's optimal decision is a function of both A1 and A2's rational decisions. In Theorem \ref{the:6}, we prove that A1's optimal decision is a convex linear combination of the two agents' rational decisions, which we call the \textit{rational decision decomposition}.
\begin{theorem}\label{the:6}
    The optimal decision $\{P_1^*(t)\}_{t\in\mathcal{T}}$ in Problem 2 is a convex linear combination of the two agents' rational decisions $\{\bar{P}_1(t)\}_{t\in\mathcal{T}}$ and $\{\bar{P}_2(t)\}_{t\in\mathcal{T}}$, i.e., 
    \begin{equation}
    P_1^*(t)=Z_1(t)\bar{P}_1(t)+[1-Z_1(t)]\bar{P}_2(t),t\in\mathcal{T},
    \label{eq:rational-decomposition-one}
    \end{equation}
    where the weight function $\{Z_1(t)\}_{t\in\mathcal{T}}$
    is given by
    \begin{equation}
        Z_1(t)=\frac{\eta \mathrm{e}^{\varrho r(T-t)}}{\eta \mathrm{e}^{\varrho r(T-t)}+\vartheta},t\in\mathcal{T}.
        \label{equ:opinion}
    \end{equation}
\end{theorem}

\begin{proof}
    By combining \eqref{eq:merton-optimal-decision}, \eqref{eq:optimal-decision-one-P2} and \eqref{eq:rational-decomposition-one}, we can obtain
    \begin{align}
        Z_1(t)&=\frac{P_1^*(t)-\bar{P}_2(t)}{\bar{P}_1(t)-\bar{P}_2(t)}=\frac{\frac{\eta\alpha_2\sigma^2\mathrm{e}^{\varrho r(T-t)}+\theta}{\eta\alpha_1\sigma^2\mathrm{e}^{\varrho r(T-t)}+\theta}\cdot\bar{P}_2(t)-\bar{P}_2(t)}{\frac{\alpha_2}{\alpha_1}\bar{P}_2(t)-\bar{P}_2(t)}\notag\\
        &=\frac{\eta\alpha_1\sigma^2\mathrm{e}^{\varrho r(T-t)}}{\eta\alpha_1\sigma^2\mathrm{e}^{\varrho r(T-t)}+\theta}=\frac{\eta \mathrm{e}^{\varrho r(T-t)}}{\eta \mathrm{e}^{\varrho r(T-t)}+\vartheta},
    \end{align}
    for all $t\in\mathcal{T}$, which is \eqref{equ:opinion}. From \eqref{equ:opinion}, the range of investment opinion is $(0,1)$, which indicates that A1's optimal decision is a convex linear combination of the two agents' rational decisions. 
\end{proof}

From \eqref{eq:rational-decomposition-one}, $\{Z_1(t)\}_{t\in\mathcal{T}}$ quantitatively describes the extent to which A1 adheres to his/her own rational decision, reflecting A1's opinion on how good his/her own rational decision is, and thereby influencing the amount of the risky asset. Therefore, we define it as A1's \textit{investment opinion}. A smaller $Z_1(t)$ indicates that A1's decision is more inclined to that of A2.

To analyze the temporal evolution of herd behaviour and its impact on A1's decision, it is essential to examine the dynamics of A1's investment opinion. In Theorem \ref{the:5}, we derive the differential equation of the investment opinion.

\begin{theorem}\label{the:5}
The differential equation of the investment opinion in \eqref{equ:opinion} is given by
\begin{align}
    \dot{Z}_1(t)&=-\varrho rZ_1(t)[1-Z_1(t)],t\in\mathcal{T},\label{eq:opinion-dynamics-one}\\
    Z_1(T)&=\frac{\eta}{\eta+\vartheta}.
    \label{eq:opinion-dynamics-one-t}
\end{align}
\end{theorem}

\begin{proof}
    According to \eqref{equ:opinion}, we have
    \begin{align}
        \dot{Z}_1(t)&=-\varrho r\frac{\eta \mathrm{e}^{\varrho r(T-t)}}{\eta \mathrm{e}^{\varrho r(T-t)}+\vartheta}+\varrho r\frac{\eta^2 \mathrm{e}^{2\varrho r(T-t)}}{[\eta \mathrm{e}^{\varrho r(T-t)}+\vartheta]^2}\notag\\
        &=-\varrho rZ_1(t)[1-Z_1(t)],
    \end{align}
    which is \eqref{eq:opinion-dynamics-one}. By substituting $T$ into \eqref{equ:opinion}, we have \eqref{eq:opinion-dynamics-one-t}.
\end{proof}

In the following, we reframe Problem 2 from the perspective of the investment opinion and show that Problem 2 is essentially equivalent to the following problem where A1 determines his/her optimal investment opinion.
\begin{align*}
    \text{\textbf{Problem 3.}}&\sup_{\{Z_1(t)\}_{t\in\mathcal{T}}\in\mathcal{U}}\mathbb{E}\phi(X_1(T))-\lambda I[Z_1]\\
    \text{s.t. \quad}&\dif X_1(t)=[rX_1(t)+vP_1(t)]\dif t+\sigma P_1(t)\dif W(t),\\
    &P_1(t)=Z_1(t)\bar{P}_1(t)+[1-Z_1(t)]\bar{P}_2(t),\\
    &X_1(0)=x_1,
\end{align*}
where $I[Z_1]=\frac{1}{2}\int_0^T\mathrm{e}^{\varrho r(t-T)}Z_1^2(t)\dif t$ and $\lambda$ is a constant.

\begin{theorem}\label{the:7}
    When $\lambda=\frac{\vartheta v^2(\alpha_1-\alpha_2)^2}{\alpha_2^2\sigma^2}$, Problem 2 and Problem 3 are equivalent, i.e., the optimal solutions $\{P_1^*(t)
    \}_{t\in\mathcal{T}}$ and $\{Z_1(t)\}_{t\in\mathcal{T}}$ are both given by \eqref{eq:optimal-decision-one} and \eqref{equ:opinion}.
\end{theorem}

\begin{proof}
    Substituting \eqref{eq:rational-decomposition-one} into Problem 2, we have
    \begin{align}
        \lambda&=\frac{\theta D[P_1\Vert\bar{P}_2]}{I[Z_1]}=\frac{\theta\int_0^T\mathrm{e}^{\rho r(T-t)}[P_1(t)-\bar{P}_2(t)]^2\dif t}{\int_0^T\mathrm{e}^{\varrho r(t-T)}Z_1^2(t)\dif t}\notag\\
        &=\frac{\theta\int_0^T\mathrm{e}^{\rho r(T-t)}[\bar{P}_1(t)-\bar{P}_2(t)]^2Z_1^2(t)\dif t}{\int_0^T\mathrm{e}^{\varrho r(t-T)}Z_1^2(t)\dif t}\notag\\
        &=\frac{\theta v^2(\alpha_1-\alpha_2)^2}{\alpha_1^2\alpha_2^2\sigma^4}=\frac{\vartheta v^2(\alpha_1-\alpha_2)^2}{\alpha_2^2\sigma^2}.
    \end{align}
    In this case, Problem 2 and Problem 3 are equivalent.
\end{proof}

\subsection{Parameters' Influence on the Optimal Decision}
With Theorem \ref{the:6} and Theorem \ref{the:5}, we can quantitatively analyze the impact of herd behaviour on A1's optimal decision. 

\subsubsection{Decay rate}
When $0\leqslant\rho<2$, we have $\dot{Z}_1(t)<0$ for all $t\in\mathcal{T}$, which indicates a monotonic decrease in A1's investment opinion over time. Consequently, A1's decision converges to that of A2 over time. This is because later deviations exert a significant impact when the decay rate is relatively small.

When $\rho=2$, we have $\dot{Z}_1(t)=0$ and $Z_1(t)=\frac{\eta}{\eta+\vartheta}$ for all $t\in\mathcal{T}$. Therefore, A1's optimal decision is a time-invariant convex linear combination of the two agents' rational decisions, i.e., $P_1^*(t)=\frac{\eta}{\eta+\vartheta}\bar{P}_1(t)+\frac{\vartheta}{\eta+\vartheta}\bar{P}_2(t)$ for all $t\in\mathcal{T}$.

When $\rho>2$, we have $\dot{Z}_1(t)>0$ for all $t\in\mathcal{T}$, which indicates that A1's investment opinion displays a monotonically increasing trend over time, while the level of convergence gradually diminishes. This is due to the relatively large decay rate of the average deviation, where later deviations have less impact. Hence, A1's decision increasingly tends toward his/her own rational decision, diverging from that of A2.

\subsubsection{Herd coefficient}
\begin{theorem}\label{the:8}
The investment opinion $Z_1(t)$ is a decreasing function of $\theta$ with $\frac{\partial Z_1(t)}{\partial \theta}<0$.
\end{theorem}

\begin{proof}
According to \eqref{eq:eta-equation}, we have
\begin{equation}
	\eta=\underline{\eta}\exp\left\{\int_0^T\frac{v^2\left(\alpha_1/\alpha_2-1\right)^2\dif t}{2\sigma^2\left[\frac{\eta}{\vartheta}\alpha_1\sigma^2\mathrm{e}^{\varrho r(T-t)}+1\right]^2}\right\}.
    \label{eq:eta-equation-2}
\end{equation}
Assume $\eta$ is a decreasing function of $\vartheta$. As $\vartheta$ increases, the left-hand side of \eqref{eq:eta-equation-2} will decrease, while the right-hand side of \eqref{eq:eta-equation-2} will increase, which is contradictory. Therefore, $\eta$ must be an increasing function of $\vartheta$, i.e., $\frac{\partial\eta}{\partial\vartheta}>0$. Let $y:=\frac{\eta}{\vartheta}$, and we have
\begin{equation}
	\eta=\underline{\eta}\exp\left\{\int_0^T\frac{v^2\left(\alpha_1/\alpha_2-1\right)^2\dif t}{2\sigma^2\left[y\alpha_1\sigma^2\mathrm{e}^{\varrho r(T-t)}+1\right]^2}\right\}.
\end{equation}
It can be easily proved that $\frac{\partial y}{\partial \eta}<0$. Therefore, we have $\frac{\partial y}{\partial\theta}=\frac{\partial y}{\partial\eta}\cdot\frac{\partial\eta}{\partial\vartheta}\cdot\frac{\partial\vartheta}{\partial\theta}<0$. According to \eqref{equ:opinion},  we have
\begin{align}
    Z_1(t)&=\frac{y\alpha_1\sigma^2\mathrm{e}^{\varrho r(T-t)}}{y\alpha_1\sigma^2\mathrm{e}^{\varrho r(T-t)}+1}\quad\mbox{and}\\
    \frac{\partial Z_1(t)}{\partial y}&=\frac{\alpha_1\sigma^2\mathrm{e}^{\varrho r(T-t)}}{[y\alpha_1\sigma^2\mathrm{e}^{\varrho r(T-t)}+1]^2}>0.
\end{align}
Therefore, we have $\frac{\partial Z_1(t)}{\partial\theta}=\frac{\partial Z_1(t)}{\partial y}\cdot\frac{\partial y}{\partial\theta}<0$. 
\end{proof}

Theorem \ref{the:8} indicates that the investment opinion is a decreasing function of the herd coefficient $\theta$. 
Therefore, the following agent's decision aligns more closely with that of the leading expert when the herd coefficient is larger.

\subsubsection{Initial wealth, excess return rate and volatility}
Next, we analyze the influence of other parameters, including the following agent's initial wealth $x_1$, the excess return rate $v$, and the volatility $\sigma$ of the risky asset. We first study the influence of these parameters on the integral constant $\eta$.
\begin{theorem}\label{the:3}
The integral constant $\eta$ is a decreasing function of $x_1$ with $\frac{\partial\eta}{\partial x_1}<0$. When $\frac{\alpha_1}{\alpha_2}\in\left[0,2\right]$, $\eta$ is a decreasing function of $v$ with $\frac{\partial \eta}{\partial v}<0$. When $\frac{\alpha_1}{\alpha_2}\in\left[1-\frac{\sqrt{3}}{3},1+\frac{\sqrt{3}}{3}\right]$, $eta$ is an increasing function of $\sigma$ with $\frac{\partial\eta}{\partial\sigma}>0$.
\end{theorem}

\begin{proof}
According to \eqref{eq:eta-equation}, we have
\begin{equation}
    \frac{\partial \eta}{\partial x_1}=-\frac{\eta\alpha_1\mathrm{e}^{rT}}{1+\eta\int_0^T\frac{\vartheta^2v^2(\alpha_1/\alpha_2-1)^2\mathrm{e}^{\varrho r(T-t)}}{\sigma^2[\eta\mathrm{e}^{\varrho r(T-t)}+\vartheta]^3}\dif t}<0.
\end{equation}
When $\frac{\alpha_1}{\alpha_2}\in\left[0,2\right]$, we have $(\frac{\alpha_1}{\alpha_2}-1)^2<1$ and
\begin{align}
    \frac{\partial \eta}{\partial v}&=\frac{\eta\left\{-\frac{vT}{\sigma^2}+\int_0^T\frac{\vartheta^2v(\alpha_1/\alpha_2-1)^2}{\sigma^2[\eta\mathrm{e}^{\varrho r(T-t)}+\vartheta]^2}\dif t\right\}}{1+\eta\int_0^T\frac{\vartheta^2v^2(\alpha_1/\alpha_2-1)^2\mathrm{e}^{\varrho r(T-t)}}{\sigma^2[\eta\mathrm{e}^{\varrho r(T-t)}+\vartheta]^3}\dif t}\notag\\
    &<\frac{\frac{\eta vT}{\sigma^2}[(\alpha_1/\alpha_2-1)^2-1]}{1+\eta\int_0^T\frac{\vartheta^2v^2(\alpha_1/\alpha_2-1)^2\mathrm{e}^{\varrho r(T-t)}}{\sigma^2[\eta\mathrm{e}^{\varrho r(T-t)}+\vartheta]^3}\dif t}<0.
\end{align}
When $\frac{\alpha_1}{\alpha_2}\in\left[1-\frac{\sqrt{3}}{3},1+\frac{\sqrt{3}}{3}\right]$, we have $(\frac{\alpha_1}{\alpha_2}-1)^2<\frac{1}{3}$ and
\begin{align}
    \frac{\partial \eta}{\partial \sigma}&=\frac{\eta\left\{\frac{v^2T}{\sigma^3}-\int_0^T\frac{\theta^2v^2(\alpha_1/\alpha_2-1)^2}{\sigma^3[\eta\alpha_1\sigma^2\mathrm{e}^{\varrho r(T-t)}+\theta]^2}\dif t\right\}}{1+\eta\int_0^T\frac{\theta^2v^2(\alpha_1/\alpha_2-1)^2\alpha_1\mathrm{e}^{\varrho r(T-t)}}{[\eta\alpha_1\sigma^2\mathrm{e}^{\varrho r(T-t)}+\theta]^3}\dif t}\notag\\
    &-\frac{\eta\int_0^T\frac{2\theta^2v^2(\alpha_1/\alpha_2-1)^2\eta\alpha_1\mathrm{e}^{\varrho r(T-t)}}{\sigma[\eta\alpha_1\sigma^2\mathrm{e}^{\varrho r(T-t)}+\theta]^3}\dif t}{1+\eta\int_0^T\frac{\theta^2v^2(\alpha_1/\alpha_2-1)^2\alpha_1\mathrm{e}^{\varrho r(T-t)}}{[\eta\alpha_1\sigma^2\mathrm{e}^{\varrho r(T-t)}+\theta]^3}\dif t}\notag\\
    &>\frac{\frac{\eta vT}{\sigma^2}[1-3(\alpha_1/\alpha_2-1)^2]}{1+\eta\int_0^T\frac{\theta^2v^2(\alpha_1/\alpha_2-1)^2\alpha_1\mathrm{e}^{\varrho r(T-t)}}{[\eta\alpha_1\sigma^2\mathrm{e}^{\varrho r(T-t)}+\theta]^3}\dif t}>0.
\end{align}

So far, we have finished the proof of Theorem \ref{the:3}.
\end{proof}

Given Theorem \ref{the:3}, we can further analyze the impact of these parameters on the investment opinion.

\begin{theorem}\label{the:4}
The investment opinion $Z_1(t)$ is a decreasing function of $x_1$ with $\frac{\partial Z_1(t)}{\partial x_1}<0$. When $\frac{\alpha_1}{\alpha_2}\in\left[0,2\right]$, $Z_1(t)$ is a decreasing function of $v$ with $\frac{\partial Z_1(t)}{\partial v}<0$. When $\frac{\alpha_1}{\alpha_2}\in\left[1-\frac{\sqrt{3}}{3},1+\frac{\sqrt{3}}{3}\right]$, $Z_1(t)$ is an increasing function of $\sigma$ with $\frac{\partial Z_1(t)}{\partial \sigma}>0$. 
\end{theorem}

\begin{proof}
According to \eqref{equ:opinion} and \eqref{eq:opinion-dynamics-one}, we can prove that
\begin{equation}
    \frac{\partial Z_1(t)}{\partial\eta}=\frac{\vartheta\mathrm{e}^{\varrho r(T-t)}}{[\eta \mathrm{e}^{\varrho r(T-t)}+\vartheta]^2}>0.
    \label{eq:wealth-opinion}
\end{equation}
Therefore, we have $\frac{\partial Z_1(t)}{\partial x_1}=\frac{\partial Z_1(t)}{\partial\eta}\cdot\frac{\partial\eta}{\partial x_1}<0$. When $\frac{\alpha_1}{\alpha_2}\in\left[0,2\right]$, we have $\frac{\partial Z_1(t)}{\partial v}=\frac{\partial Z_1(t)}{\partial\eta}\cdot\frac{\partial\eta}{\partial v}<0$. When $\frac{\alpha_1}{\alpha_2}\in\left[1-\frac{\sqrt{3}}{3},1+\frac{\sqrt{3}}{3}\right]$, we have $\frac{\partial Z_1(t)}{\partial\sigma}=\frac{\partial Z_1(t)}{\partial\eta}\cdot\frac{\partial\eta}{\partial\sigma}>0$. 
\end{proof}

From Theorem \ref{the:4}, we can draw the following conclusions.

The following agent's investment opinion is a decreasing function of his/her initial wealth $x_1$, indicating that when the following agent has a higher initial wealth, the investment opinion $Z_1(t)$ is smaller, and from \eqref{eq:rational-decomposition-one}, his/her decision tends to converge more to that of the leading expert. This can also be intuitively explained using the diminishing margin feature of the utility function. As the initial wealth increases, the additional utility gained from terminal wealth diminishes. Therefore, the following agent pays more attention to reducing the disparity from the leading expert's decision, causing the convergence of the two agents' decisions.

Similarly, when the two agents' risk aversion coefficients are close, a higher excess return rate $v$ and lower volatility $\sigma$ result in a smaller investment opinion from Theorem \ref{the:4}. Therefore, from \eqref{equ:opinion}, the following agent's decision converges more to that of the leading expert. This can be intuitively explained as follows. With a larger $v$ and a smaller $\sigma$, the risky asset achieves a higher expected utility of terminal wealth. Thus, from the diminishing margin feature of the utility function, the following agent's decision tends to converge more to that of the leading expert.

\section{Numerical Experiments}
\label{sec:experiment}

In this section, we conduct numerical experiments to validate our analyses in Section \ref{sec:solution} and Section \ref{sec:analysis}.

\begin{figure}[!t]
\centering
\includegraphics[width=\linewidth]{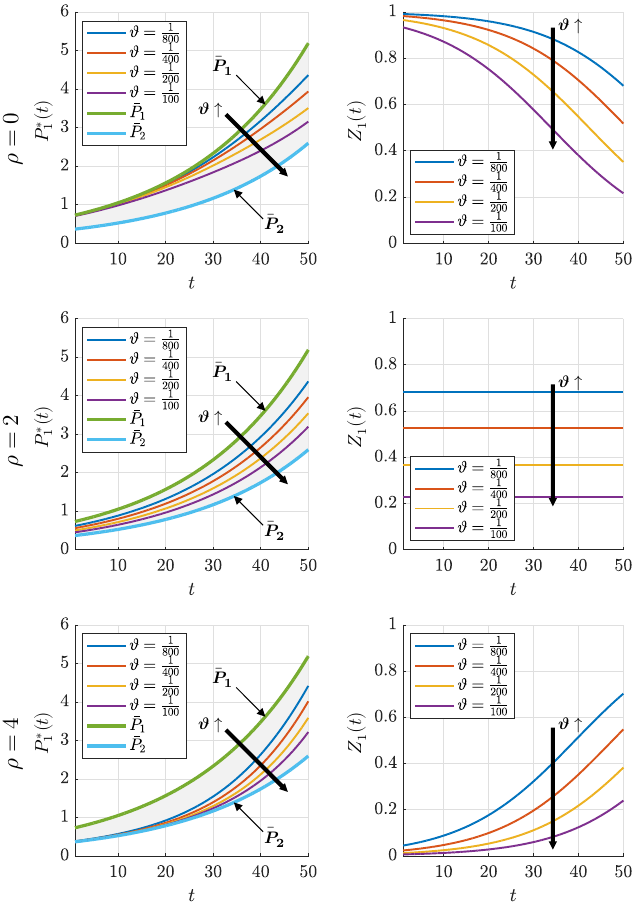}
\caption{Optimal decisions and investment opinions with different decay rates.}
\label{fig:fig2}
\end{figure}

\subsection{Parameter Settings}
We choose a temporal span covering fifty years from January 1972 to December 2022, i.e., $T=50$. During this investment period, we collect daily closing prices of the Dow Jones Industrial Average to represent the prices of the risky asset. Utilizing parameter estimation of the geometric Brownian motion, we determine the expected return $\mu$ to be $0.07$ and volatility $\sigma$ to be $0.17$ for the risky asset. We use the daily average of the interest rates of U.S. 1-Year Treasury Bills in 2022 and 2023 to represent the risk-free interest rate $r$, which is approximately equal to $0.04$. From the prior work in \cite{yuen2001estimation}, we set $\alpha_1=0.2$ and $\alpha_2=0.4$ as the risk aversion coefficients for the following agent and leading expert, respectively. We set the modified herd coefficient $\vartheta$ to values of $\frac{1}{800}$, $\frac{1}{400}$, $\frac{1}{200}$, and $\frac{1}{100}$. We set the following agent's initial wealth as $X_{1,0}=0$. We observe the same trend for other values of the parameters. 

\subsection{Optimal Decision}
We consider three cases with different decay rates: $\rho=0$, $\rho=2$, and $\rho=4$. Experiment results are shown in Fig. \ref{fig:fig2}. It can be seen that A1's optimal decision consistently remains within the region bounded by the rational decisions of the two agents $\{\bar{P}_1(t)\}_{t\in\mathcal{T}}$ and $\{\bar{P}_2(t)\}_{t\in\mathcal{T}}$. Furthermore, the optimal decision progressively approaches A2's rational decision $\{\bar{P}_2(t)\}_{t\in\mathcal{T}}$ as the herd coefficient increases. When $\rho=0$, the optimal decision increasingly converges towards A2's rational decision $\{\bar{P}_2(t)\}_{t\in\mathcal{T}}$, and the investment opinion decreases over time. When $\rho=2$, the optimal decision maintains a constant distance from the rational decisions of the two agents, and the investment opinion remains constant. When $\rho=4$, the optimal decision increasingly converges towards A1's rational decision $\{\bar{P}_1(t)\}_{t\in\mathcal{T}}$, and the investment opinion increases over time. These observations agree with our analyses of the optimal decision and investment opinion in Theorem \ref{the:6}.

\subsection{Parameters' Influence on the Optimal Decision}
Next, we study the parameters' influence on A1's optimal decision.
As shown in Fig. \ref{fig:fig2}, when $\vartheta$ increases, the investment opinion $Z_1(t)$ decreases. Therefore, $Z_1(t)$ is a decreasing function of the modified herd coefficient $\vartheta$. Because $\theta=\alpha_1\sigma^2\vartheta$, $Z_1(t)$ is also a decreasing function of the herd coefficient $\theta$, which validates the correctness of Theorem \ref{the:8}. 

Keeping other parameters constant and ensuring $\frac{\alpha_1}{\alpha_2} \in \left[1-\frac{\sqrt{3}}{3},1+\frac{\sqrt{3}}{3}\right]$, we plot the curves of the integral constant $\eta$ with respect to the initial wealth $x$, excess return rate $\mu$, and volatility $\sigma$. As shown in Fig. \ref{fig:fig3}, $\eta$ increases with the decrease of $x$ and $v$, and increases with the increase of $\sigma$, thereby validating the correctness of Theorem \ref{the:3}. 

\section{Conclusion}
\label{sec:conclusion}
In this paper, we formulate a dual-agent optimal investment problem considering herd behaviour and introduce the average deviation term in the traditional Merton problem's objective functional to measure the distance between the two agents' decisions. We obtain the analytical solution using the variational method and quantitatively analyze the impact of herd behaviour on the following agent's optimal decision using the rational decision decomposition. Furthermore, we introduce the concept of investment opinion to quantify the following agent's preference for his/her own rational decision over that of the leading expert. We validate our analyses through numerical experiments on real stock data. 

\begin{figure}[!t]
\centering
\subfloat[Initial wealth]{\includegraphics[width=0.33\linewidth]{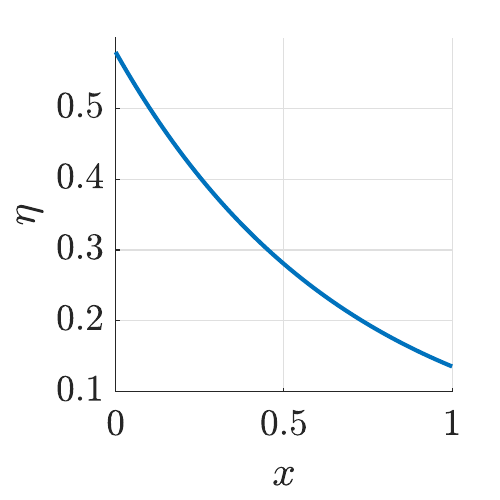} \label{fig:fig3b}}
\subfloat[Excess return rate]{\includegraphics[width=0.33\linewidth]{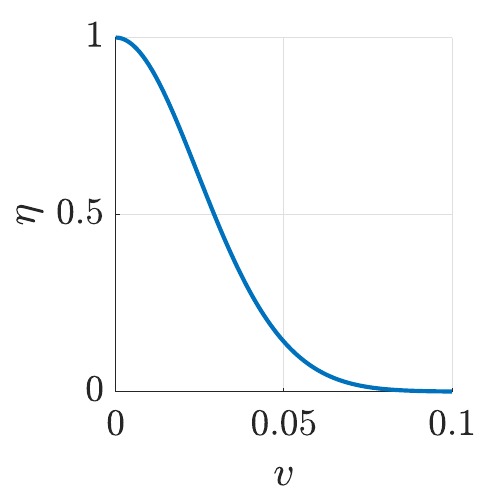} \label{fig:fig3c}}
\subfloat[Volatility]{\includegraphics[width=0.33\linewidth]{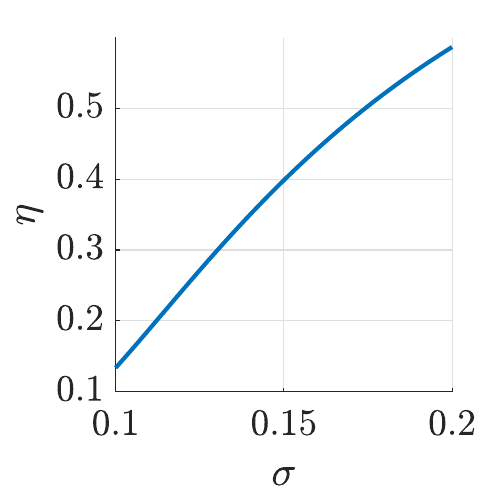} \label{fig:fig3d}}
\caption{Parameters' influence on the integral constant $\eta$.}
\label{fig:fig3}
\end{figure}

\bibliography{bibliology}
\bibliographystyle{IEEEtran}

\end{document}